\documentclass[conference]{IEEEtran}

\usepackage{graphicx, color} 
\usepackage{epsfig} 
\usepackage{amsmath, bbm, bm} 
\usepackage{amsfonts, amssymb}
\usepackage{mathrsfs}
\usepackage{bbm}
\usepackage{tikz}
\usepackage{algorithm}
\usepackage{subcaption}
\usepackage{algpseudocode}
\usepackage{url}

\newcommand {\N} {{\rm I\kern-1.5pt N}}
\newcommand {\R} {{\rm I\kern-2.5pt R}}
\newcommand {\C} {{\rm I\kern-5pt C}}

\newtheorem{lemma}{Lemma}

\newtheorem{theorem}{Theorem}
\newtheorem{remark}{Remark}

\newcommand{\beqa}{\begin{eqnarray}}
\newcommand{\eeqa}{\end{eqnarray}}
\newcommand{\beqan}{\begin{eqnarray*}}
\newcommand{\eeqan}{\end{eqnarray*}}
\newcommand{\beq}{\begin{equation}}
\newcommand{\eeq}{\end{equation}}
\newcommand{\prf}{\noindent {\bf Proof}\ \ \ }

\newcommand{\bfl}{\begin{flushleft}}
\newcommand{\efl}{\end{flushleft}}
\newcommand{\bgnv}{\begin{verbatim}}
\newcommand{\ednv}{\end{verbatim}}

\newcommand{\myb}{\hspace{-0.1in}}
\newcommand{\myeq}{& \hspace{-0.1in} = & \hspace{-0.1in}}

\newcommand{\lb}{\nonumber \\}
\newcommand{\myarr}{\begin{array}{lll}}

\newcommand{\myleq}{& \myb \leq & \myb}
\newcommand{\myl}{& \myb < & \myb}
\newcommand{\bitem}{\begin{itemize}}
\newcommand{\eitem}{\end{itemize}}
\newcommand{\benum}{\begin{enumerate}}
\newcommand{\eenum}{\end{enumerate}}

\newcommand{\myhb}{\hspace{-0.3in}}

\newcommand{\myskip}{\\ \vspace{-0.1in}}

\def\QED{~\rule[-1pt]{5pt}{5pt}\par\medskip}
\newenvironment{proof}{{\bf Proof: \ }}{ \hfill \QED}

\usepackage{framed}
\usepackage[framemethod=TikZ]{mdframed}
\mdfdefinestyle{theoremstyle}{backgroundcolor=black!7,roundcorner=5pt,outerlinewidth=1.1}

\begin{document}

\title{Miniature Robot Path Planning for 
	Bridge Inspection:
	Min-Max Cycle Cover-Based Approach}

\author{\IEEEauthorblockN{Michael Lin and
Richard J. La\thanks{The authors are with the Department
of Electrical and Computer Engineering and the 
Institute for Systems Research at the University of
Maryland, College Park.}
\thanks{This work was supported in part by an NSF CPS
grant}}
\IEEEauthorblockA{Department
of Electrical and Computer Engineering \\
and the 
Institute for Systems Research \\
University of
Maryland, College Park \\
\{mlin1025, hyongla\}@umd.edu}
}

\maketitle

\begin{abstract}

We study the problem of planning the deployments 
of a group of mobile robots. While the problem 
and formulation can be used for many different 
problems, here we use a bridge inspection as the 
motivating application for the purpose of 
exposition. The robots are initially stationed 
at a set of depots placed throughout the bridge. 
Each robot is then assigned a set of sites on the 
bridge to inspect and, upon completion, must
return to the same depot where it is stored. 

The problem of robot planning 
is formulated as a rooted min-max 
cycle cover problem, in which the vertex set 
consists of the sites to be inspected and robot 
depots, and the weight of an edge captures either 
(i) the amount of time needed to travel from one 
end vertex to the other vertex or (ii) the 
necessary energy expenditure for the travel. In the
first case, the objective function is the total 
inspection time, whereas in the latter case, it 
is the maximum energy expenditure among all 
deployed robots. We propose a novel
algorithm with approximation ratio of $5
+ \epsilon$, where $0<\epsilon<1$. In addition, 
the computational complexity of the proposed 
algorithm is shown to be $O\big( n^2+2^{m-1} n 
\log(n+k) \big)$, where $n$ is the number of 
vertices, and $m$ is the number of depots.  
\end{abstract}


\section{Introduction}
	\label{sec:Introduction}

With aging infrastructure, ensuring the safety 
of existing civil structures, such as 
bridges, roads and tunnels, is becoming an 
important societal challenge. Inadequate 
monitoring of infrastructure can result 
in major incidents, such as the collapse of
bridges, e.g., the failure of Ponte Morandi
bridge in Italy in August 2018, which killed 43 
people. According to a 2018 U.S. Federal Highway 
Administration (FHWA) report, 
more than 47,000 bridges 
are deemed to be in ``poor" condition out of 
approximately 616,000
bridges, and nearly a half of all
bridges are found to be in ``fair" condition
\cite{FHWA}. 

Unfortunately, many segments of a bridge are not 
easily accessible, making it difficult for human 
inspectors to perform frequent inspections. 
As a result, many bridges are not inspected 
frequently enough to maintain their structural 
health and safety, which is reflected 
in the U.S. FHWA report, thereby raising
the possibility of suffering another
major bridge collapse in the future.

Rapid advances in robotics technologies make it 
possible to employ small mobile robots to help 
with the inspection of different types of 
structures, including bridges. These robots
will likely be battery powered to improve their
mobility, thereby limiting their ranges and 
tasks that they can perform before their
battery needs to be recharged. For this reason, 
in order to complete a bridge inspection as 
quickly as possible, it is important to take 
into account their energy constraints when 
employing the robots for inspection. 

\begin{figure}[h]
\centerline{
\includegraphics[width=3.5in]{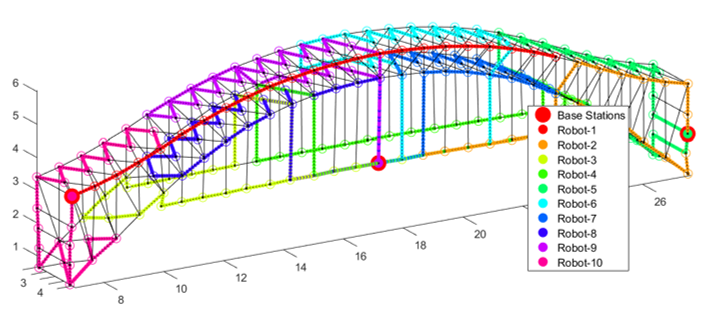}
}
\caption{Bridge inspection by mobile robots.}
\label{fig:bridge}
\end{figure}

We study a robot planning problem in which a 
group of battery-powered mobile robots are 
stored and recharged at a set of depots, and 
are utilized for inspecting a set of sites.
These sites could, for instance, 
represent various points on a bridge that need 
to be inspected (e.g., joints) by a robot.
This is illustrated in Fig.~\ref{fig:bridge}.  In 
the figure, vertices represent the set of 
sites on a bridge to be inspected by the robots, 
and edges show paths robots
can take to move between various points
on the bridge.

In our problem, we are interested in determining, 
for each robot, (i) a depot where it is to be 
stored (and recharges its battery) and 
(ii) a set of sites for the robot to inspect. 
We require that, upon completing the inspection 
of all assigned sites, the robot must return to 
the same depot where it is stored.

The problem is formulated as a rooted min-max
cycle cover problem. Each cycle in the cycle
cover, which is rooted at a depot, 
is assigned to a robot and determines a 
subset of sites that the robot must inspect
as well as the depot at which the robot is
to be stationed. We propose a new algorithm
for the rooted min-max cycle cover problem with 
approximation ratio of $5 + \epsilon$
($0 < \epsilon < 1$). 

The rest of the paper is organized as follows.
Section~\ref{sec:Related} summarizes some
of the most closely related studies in the
literature. Section \ref{sec:Model} 
presents the model and
formulation of our problem. Our proposed
algorithm is described in Sections 
\ref{sec:Preliminaries} and 
\ref{sec:Complete}. We discuss the complexity
of the proposed algorithm in Section
\ref{sec:Complexity}.

\section{Related Work}
	\label{sec:Related}

The well-known traveling salesman problem (TSP) is 
NP-hard. Since the TSP is a special case with one 
depot and one agent, the general rooted min-max
cycle cover problem is also NP-hard. 
For this reason, researchers
proposed approximation algorithms for related
problems over the years. 

Even et al. \cite{even2004min} studied the 
min-max cycle cover problem in the context 
of nurse station location problem. 
They proposed algorithms for both rooted
and unrooted (or rootless) min-max tree 
cover problems 
with approximation ratio of $4+\epsilon$
($\epsilon > 0$). 
This provides an $8+\epsilon$ 
approximation algorithm for the rooted 
min-max cycle cover problem.

In a closely related study, Arkin et al. 
\cite{arkin2006approximations} 
also provided a $4+\epsilon$ approximation 
algorithm for the unrooted min-max tree 
cover problem.  
Khani and Salavatipour \cite{khani2014improved} 
improved the approximation ratio to 
$3+\epsilon$ for unrooted min-max tree
cover problem, which in turn 
yields an approximation ratio of $6+\epsilon$ 
for the unrooted min-max cycle cover problem. 

Rather than starting with a tree cover problem, 
Jorati \cite{jorati2013approximation} directly
studied both rooted and unrooted min-max
cycle cover problems and proposed
algorithms with approximation ratio of 
$5\frac{1}{3} + \epsilon$ and 
$7+\epsilon$ for unrooted 
and rooted problem, respectively. Similarly, 
in an independent study,  
Xu et al. \cite{xu2015approximation} 
investigated the same cycle cover problems and 
proposed algorithms with 
approximation ratio of $5\frac{1}{3}+\epsilon$
and $6\frac{1}{3}+\epsilon$
for unrooted and (uncapacitated) rooted
min-max cycle cover
problems, respectively.\footnote{The authors of 
\cite{xu2015approximation} also
studied the capacitated rooted
min-max cycle cover problem and proposed
an algorithm with approximation
ratio of $7+\epsilon$.}
In addition, Yu and Liu 
\cite{yu2016improved} proposed 
algorithms with improved approximation 
ratio of $5+\epsilon$ and $6+\epsilon$ for
unrooted and rooted min-max cycle cover
problems, respectively, 
by utilizing the well-known 
Christofides algorithm \cite{christofides1976worst} 
for the TSP problem.
 
We point out that there are other studies on special 
cases of the cycle cover problems with 
better approximation ratios, e.g.,
\cite{frederickson1976approximation},  
\cite{farbstein2015min}, 
\cite{nagamochi2005approximating}, 
\cite{nagamochi2007approximating}, 
\cite{karakawa2009minmax}. 
For example, Frederickson et al. 
\cite{frederickson1976approximation} proposed a 
$2\frac{1}{2}+\epsilon$ approximation algorithm 
for single depot case.
Finally, Xu and 
Wen \cite{xu2010approximation}
proved that, unless $P = NP$, 
there exists no polynomial time 
$(1\frac{3}{17} - \epsilon)$-approximation algorithm
for the min-max cycle cover problem with a single
root. This result is generalized by Xu et 
al. \cite{xu2012approximation} who showed that there
does not exist a polynomial time algorithm for 
the unrooted and rooted min-max 
cycle cover problems with an approximation ratio 
less than $1\frac{1}{3}$ unless $P=NP$.

It is worth mentioning that, although our rooted 
min-max cycle cover formulation is inspired by 
\cite{xu2015approximation} and \cite{yu2016improved}, 
our model is slightly different: they only 
required the union of cycles to cover all vertices, except 
for depots. In other words, the depots need not be
covered by the cycle cover. In our problem, however, 
we require that all vertices, including depots, be 
covered by the cycle cover.

Instead of investigating the unrooted 
min-max cycle cover 
problem and applying the algorithm on the rooted 
version as in many, if not most, of related
studies, we directly 
tackle the rooted min-max cycle cover problem. 
We propose an approximation algorithm that runs 
in polynomial time with a fixed number of depots. 
The approximation ratio of $5+\epsilon$ is 
comparable to the state-of-the-art algorithm 
for unrooted problem \cite{yu2016improved}. 
Moreover, even though our 
formulation is somewhat different, our approximation 
ratio is better than the previous best algorithm for 
rooted cycle cover problem \cite{yu2016improved}.

\section{Model and Formulation}
	\label{sec:Model}

We formulate the problem of robot planning 
as a rooted min-max cycle cover problem on a 
complete undirected graph $G = (V, E)$: the 
vertex set $V$ consists of both (a) 
the sites to be inspected by the robots 
and (b) the depots where
the robots are stored, 
and each undirected edge $e$ in $E$ has a 
weight associated with it. The goal of 
the problem is to find a set of cycles subject 
to following two constraints:
(i) each cycle is rooted
at a depot in that it starts 
and ends at the same depot, and 
(ii) the union of 
all cycles covers all vertices in $V$. 
The interpretation is that each cycle found
in the problem is assigned to a unique robot
and determines the set of sites to be visited
by the robot as well as the depot at which
the robot will be stored. 

The weights assigned to the edges depend on our
objective. We consider two different
choices of weights. Obviously, it is possible to 
take a weighted sum of the two choices. 

\begin{itemize}

    \item {\bf Case~I:} The weight of an edge 
	models the minimum amount of time needed to 
	travel from one end vertex to the other 
	vertex. In this case, the overall cost 
	for a robot is equal to the total travel 
	time for the robot to visit all assigned 
	sites in the cycle and come back to the 
	depot. The objective of 
	our problem is to {\em minimize the total 
	amount of time needed for inspection}, 
	which is equivalent to minimizing the 
	maximum cost among all 
	robots.\footnote{Here, we implicitly
	assume that the amount of time it takes
	to inspect the sites is negligible 
	compared to the travel time. However, 
	the case with non-zero weights associated
	with vertices can be handled by constructing
	a new graph with only edge weights as 
	described in \cite{xu2012approximation}.}

    \item {\bf Case~II:} In the second case,
	the edge weight indicates the necessary 
	energy expenditure for the travel from 
	one end vertex to the other. The goal is 
	then to {\em minimize the maximum energy 
	expenditure among all robots} for the given 
	assignments. 

\end{itemize}

The second formulation allows us to determine 
whether or not the battery-powered robots can 
perform the inspection without having to 
recharge; if the optimal value 
of the optimization problem is larger than 
the amount of battery energy available to the 
robots, it suggests that recharging is necessary
for some robots.

These edge weights can be obtained from the
weights we can estimate from the physical 
structure. For example, Fig.~\ref{fig:bridge}
tells us the available paths between various
points on the bridge. Once we estimate the
weights of these available paths on the bridge, 
we can use a well-known algorithm, such as 
Dijkstra's algorithm~\cite{Algorithm}, 
to compute the weights
of the shortest paths between any pair of 
vertices in $V$.

\subsection{Rooted min-max cycle cover problem}

Suppose (i) $G = (V, E)$ is a complete 
undirected graph described earlier, 
(ii) $w: E \to \R_+ := [0, \infty)$ 
is an edge weight 
function, (iii) $D \subseteq V$ is a 
set of depots (where the robots are stored), 
and (iv) $k$ is a positive integer (which is
equal to the number of robots in our problem).
Given a subgraph $G'$ 
of $G$, e.g., a cycle or a tree, let 
${\cal E}(G')$ and 
${\cal V}(G')$ be the 
set of edges and the set of vertices, 
respectively, in $G'$.

Denote by ${\cal C}^k$ the set of all possible 
edge-disjoint, rooted cycle covers of $V$ 
with at most $k$ cycles. 
In other words, an element $\mathscr{C} 
= \{ C_1, \ldots, C_q\}$ of ${\cal C}^k$ 
consists of $q$ cycles in $G$ satisfying 
the following:
\begin{enumerate}
\item[c1.] $q \leq k$; 
    
\item[c2.] ${\cal E}(C_i) \cap {\cal E}(C_j) = 
    \emptyset$ for all $i \neq j$; 
    
\item[c3.] for every cycle $C_j \in \mathscr{C}$, 
    $|{\cal V}(C_j) \cap D| = 1$; and 
    
\item[c4.] $\cup_{j=1}^q {\cal V}(C_j) = V$.
\end{enumerate}
Given a cycle cover
$\mathscr{C}$ in ${\cal C}^k$, we denote 
the number of cycles in $\mathscr{C}$ by 
$\xi(\mathscr{C})$. 
 
The rooted min-max cycle 
cover problem we adopt is the following 
optimization problem:\footnote{In some cases, 
it may make sense to limit the number of
robots that can be stored at each robot. 
This would give rise to a {\em capacitated}
rooted min-max cover problem. Here, we 
assume no such constraints on depot sizes.}  
\myskip

\begin{mdframed}[style=theoremstyle]
\noindent
\underline
{\bf Rooted Min-Max Cycle Cover Problem}
\beqa
{\rm minimize}_{\mathscr{C} \in {\cal C}^k}
& & \max_{j = 1,\ldots,\xi(\mathscr{C})} 
	\Big(
	\sum_{e \in {\cal E}(C_j)} w(e)
	\Big) 
	\label{eq:objective} 
\eeqa
\end{mdframed}
where, with a little abuse of notation, 
$C_j$ denotes the $j$-th cycle
in $\mathscr{C}$. 
Therefore, the goal of the rooted min-max 
cycle cover problem is to find $q$ 
edge-disjoint cycles $\mathscr{C} 
= \{C_1, \ldots, C_{q}\}$
such that (a) $q \leq k$, 
(b) each cycle contains exactly one depot in $D$, 
(c) the union of these cycles includes all
vertices in $V$, and 
(d) the maximum weight of cycles is minimized.
As mentioned earlier, a feasible 
solution we obtain from 
the optimization problem in \eqref{eq:objective}
also determines the (minimum) number of robots 
that will be stored at each depot. 

Throughout the paper, 
we assume that the edge weight function $w$ 
is a metric and, hence, satisfies symmetry 
and triangle inequality: for all distinct 
vertices $v_1, v_2, v_3 \in V$, we have
\beqan
w(v_1,v_2) 
\leq w(v_1,v_3) + w(v_3,v_2).
\eeqan
This is a natural assumption for both choices
of edge weights discussed earlier. 

\begin{remark}
Recall that the rooted min-max cycle
cover problem in \eqref{eq:objective} 
permits only edge-disjoint cycles. 
Our problem of robot planning does not 
explicitly require that the cycles be 
edge-disjoint. However, using the assumption 
that the edge weight function $w$ is a 
metric, one can easily show that the
following holds: given a feasible 
non-edge disjoint cycle cover $\mathscr{C}$, 
we can construct an edge-disjoint 
cycle cover $\mathscr{C}^{{\rm ed}}$ such 
that the maximum cycle weight of 
$\mathscr{C}^{{\rm ed}}$ is less than or 
equal to that of $\mathscr{C}$. For this
reason, without loss of generality, we
can focus on edge-disjoint cycle
covers in the rooted min-max cycle cover
problem in \eqref{eq:objective}. 
\end{remark}

\begin{remark}
It is clear from the constraints that, if 
$k < |D| =: m$, 
then there is no feasible solution; 
since each cycle can include only one depot
(constraint c3), 
at most $k$ depots can be contained in the 
union $\cup_{j=1}^k C_j$ and, as a 
result, constraint c4
cannot be satisfied. Although this issue can
be dealt with by changing the formulation 
slightly, we 
assume that there are more robots than the
depots, i.e., $k \geq m$. This is a reasonable
assumption for our problem as the number
of depots is expected to be small with each
depot housing many robots.
\end{remark}

\begin{remark}
From the problem formulation in 
\eqref{eq:objective}, the cycle cover
we are looking for must cover not only 
all the sites of interest, but also all 
depots in $D$. In other words, we require
that at least one robot be stationed at each
depot. Although we make this assumption 
explicit, this is likely to be satisfied in practice 
even without making it explicit; depots should be
spread out across the bridge for storage and 
recharging, and sites close to each depot
should be assigned to robot(s) stationed
at the depot. Otherwise, the depot should 
be removed. 
\end{remark}


\section{Preliminaries: Key Steps of the 
	Proposed Algorithm}
	\label{sec:Preliminaries}

The proposed algorithm (Algorithm
\ref{alg:complete} in Section
\ref{sec:Complete}) consists of several
steps that we discuss in detail 
in this section. 
The input to the algorithm comprises the 
information for the rooted min-max cycle
cover problem: 
(a) a complete undirected graph $G=(V, E)$, 
(b) a metric edge weight function $w$, 
(c) the number of cycles $k$, and 
(d) the depot set $D$.

\subsection{Step 1: Construction
	of a rooted spanning forest $F^*$}
	\label{stepA}

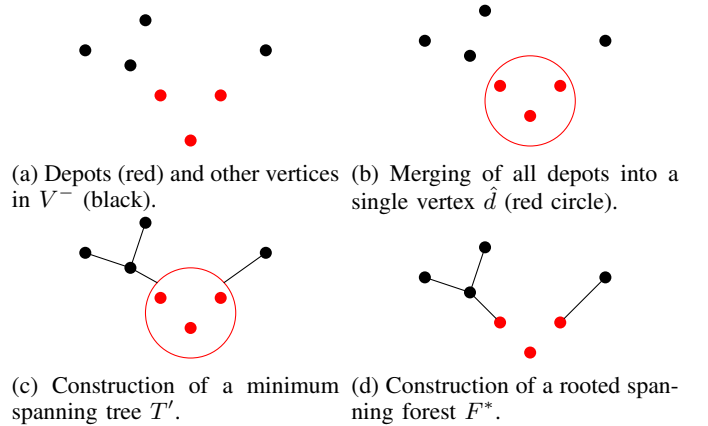
\begin{figure}[h]
    \centering
    \begin{subfigure}[t]{0.24\textwidth}
        \centering
        \begin{tikzpicture}[scale=0.4, every node/.style={transform shape}]
            \fill[color=red] (0,0.5) node (v1) {} circle (.2);
            \fill[color=red] (1,2) node (v1) {} circle (.2);
            \fill[color=red] (-1,2) node (v1) {} circle (.2);
            \fill[color=black] (2.5,3.5) node (v1) {} circle (.2);
            \fill[color=black] (-2,3) node (v1) {} circle (.2);
            \fill[color=black] (-1.5,4.5) node (v1) {} circle (.2);
            \fill[color=black] (-3.5,3.5) node (v1) {} circle (.2);
        \end{tikzpicture}
        \caption{Depots (red) and other vertices in $V^{-}$ (black).}
    \end{subfigure}%
    ~ 
    \begin{subfigure}[t]{0.24\textwidth}
        \centering
        \begin{tikzpicture}[scale=0.4, every node/.style={transform shape}]
            \draw[color=red] (0,1.5) circle (1.5);
            \fill[color=red] (0,1) node (v1) {} circle (.2);
            \fill[color=red] (1,2) node (v1) {} circle (.2);
            \fill[color=red] (-1,2) node (v1) {} circle (.2);
            \fill[color=black] (2.5,3.5) node (v1) {} circle (.2);
            \fill[color=black] (-2,3) node (v1) {} circle (.2);
            \fill[color=black] (-1.5,4.5) node (v1) {} circle (.2);
            \fill[color=black] (-3.5,3.5) node (v1) {} circle (.2);
        \end{tikzpicture}
        \caption{Merging of all depots into 
        	a single vertex $\hat{d}$ (red 
        	circle).}
    \end{subfigure}
        \begin{subfigure}[t]{0.24\textwidth}
        \centering
        \begin{tikzpicture}[scale=0.4, every node/.style={transform shape}]
            \draw[color=red] (0,1.5) circle (1.5);
            \fill[color=red] (0,1) circle (.2);
            \fill[color=red] (1,2) circle (.2);
            \fill[color=red] (-1,2) circle (.2);
            \fill[color=black] (2.5,3.5) node (v4) {} circle (.2);
            \fill[color=black] (-2,3) node (v2) {} circle (.2);
            \fill[color=black] (-1.5,4.5) node (v1) {} circle (.2);
            \fill[color=black] (-3.5,3.5) node (v3) {} circle (.2);
            \draw (v1) -- (v2);
            \draw (v3) -- (v2);
            \draw (v4) -- (1.1,2.5);
            \draw (v2) -- (-1.1,2.5);
        \end{tikzpicture}
        \caption{Construction of a minimum 
        	spanning tree $T'$.}
    \end{subfigure}%
    ~ 
    \begin{subfigure}[t]{0.24\textwidth}
        \centering
        \begin{tikzpicture}[scale=0.4, every node/.style={transform shape}]
            \fill[color=red] (0,1) circle (.2);
            \fill[color=red] (1,2) node (v5) {} circle (.2);
            \fill[color=red] (-1,2) node (v6) {} circle (.2);
            \fill[color=black] (2.5,3.5) node (v4) {} circle (.2);
            \fill[color=black] (-2,3) node (v2) {} circle (.2);
            \fill[color=black] (-1.5,4.5) node (v1) {} circle (.2);
            \fill[color=black] (-3.5,3.5) node (v3) {} circle (.2);
            \draw (v1) -- (v2);
            \draw (v3) -- (v2);
            \draw (v5) -- (v4);
            \draw (v2) -- (v6);
        \end{tikzpicture}
        \caption{Construction of a rooted 
        	spanning forest $F^*$.}
    \end{subfigure}
    \caption{Illustration of Step 1 
    for constructing a rooted spanning
    forest.}
    \label{fig:mforest}
\end{figure}

The first step generates a rooted forest with 
$m$ trees and consists of the following two 
steps:\footnote{The forests we construct are 
in fact tree covers for $G$. For consistency, 
we shall refer to them as forests in the 
remainder of the paper.}

{\bf Step 1-i:} 
First, we collapse all depots in $D$ into a 
single node, which we denote by $\hat{d}$. 
This is shown in Fig.~\ref{fig:mforest}(b).
Second, 
we create an edge from $\hat{d}$ to every 
vertex $v$ in $V \setminus D =: V^-$ 
with an edge weight
\beqan
w(\hat{d},v) \myeq \min_{d \in D}w(d,v).
\eeqan
We denote the complete undirected graph with 
vertex set $V' := V^- \cup 
\{\hat{d}\}$ by $G' = (V', E')$. 
Third, we compute a minimum spanning tree 
$T'$ of $G'$, using Prim's algorithm
\cite{Algorithm}, as shown 
in Fig.~\ref{fig:mforest}(c). 

{\bf Step 1-ii:} 
We uncouple the node $\hat{d}$ back into the 
original depots to create a forest $F^*$
as follows: first, let $F^* = T'$. Second, 
add the depots in $D$ and replace each 
edge $(v, \hat{d})$, $v \in V^-$, with 
an edge ($d^\star, v)$, 
where $d^\star \in \arg\min_{d \in D}
w(d, v)$. When $| \arg\min_{d \in D} w(d, v)| 
> 1$, we arbitrarily choose a depot in 
the set. Finally, we remove $\hat{d}$. 
The resulting rooted forest $F^*$ has exactly 
$m$ trees in it, which are denoted by
$T^*_1, \ldots, T^*_m$. This is illustrated
in Fig.~\ref{fig:mforest}(d).

\subsection{Step 2: Virtual 
	Forest ($F$) and an Exhaustive Search}
\label{stepB}

The second step takes $F^*$ generated in
Step 1 as the input
and produces a set of forests with trees
that are created by merging
the trees in $F^*$ in different ways. 
In particular, it considers $2^{m-1}$
different possible ways in which 
the $m$ trees in $F^*$ can be connected
with each other. The reason for this is
to ensure that we consider at least one 
forest that will lead to a rooted cycle cover 
with a provable upper bound on the maximum 
cycle weight (Steps 3 and 4).
Specifically, we want to make sure that 
the forest described below is considered
in Steps 3 and 4 through an exhaustive
search. 

{\bf Step 2-i:} 
Consider the construction of a spanning 
tree using Algorithm~\ref{algo:ST} below, 
starting with the forest $F^*$.

\begin{algorithm}[h]
\caption{Construction of Edge Set 
	$E^\dagger$}
\label{algo:ST}
\begin{algorithmic}[1]
\State {Let $F_{{\rm tmp}} = F^*$ and $E^\dagger 
= \emptyset$}
\While{there is more than one tree in $F_{{\rm tmp}}$}
	\State {Find an edge $e^c \in E$ that joins two 
	distinct trees $T_1$ and $T_2$ in 
	$F_{{\rm tmp}}$ with the smallest weight}
	\State {Add $e^c$ to $E^\dagger$}
	\State {Remove $T_1$ and $T_2$ from $F_{{\rm tmp}}$,
	and add the new tree formed after connecting
	$T_1$ and $T_2$ with $e^c$ to $F_{{\rm tmp}}$}
\EndWhile
\end{algorithmic}
\end{algorithm}
Note that the edge set $E^{\dagger}$ produced
by Algorithm~\ref{algo:ST} contains $m-1$
edges that connect the $m$ trees in $F^*$
into a single spanning tree. 

\paragraph{Forest $F_{{\rm opt}}$}
Suppose that $\mathscr{C}^* = \{C_1^*, \ldots, 
C_{q^*}^*\}$, where $q^* \leq k$,  is a solution to 
\eqref{eq:objective}, i.e., an optimal cycle cover,   
and the maximum weight of cycles in 
$\mathscr{C}^*$ is $\lambda^*$.
Note that $\lambda^*$ is the optimal value
of \eqref{eq:objective}. 
Throughout the remainder of the paper, with 
a little abuse of notation, we denote the
total weight of edges in a subgraph 
$\tilde{G}$, such as trees and cycles, 
by $w(\tilde{G})$.   

Assume for now that the optimal value
$\lambda^*$ and the optimal cycle cover
$\mathscr{C}^*$ are known. 
We classify a tree $T$ into three categories 
on the basis of the optimal value $\lambda^*$:
\begin{enumerate}
\item[T1.] Heavy tree : $w(T) \geq \lambda^*$.

\item[T2.] Light tree : $w(T) < \lambda^*$ 
and there exists a cycle $C \in \mathscr{C}^*$
such that ${\cal V}(C) \subseteq {\cal V}(T)$, i.e., 
$T$ contains all the vertices in at least one
cycle in the optimal cycle cover $\mathscr{C}^*$. 

\item[T3.] Bad tree : $w(T) < \lambda^*$ 
and there is no cycle $C \in \mathscr{C}^*$
such that ${\cal V}(C) \subseteq {\cal V}(T)$.
\end{enumerate}

Based on this classification of trees,
starting with the forest $F^*$, we perform
a procedure in Algorithm~\ref{algo:ElimBad} 
to eliminate bad trees and construct a new 
forest free of bad trees.

\begin{algorithm}[h]
\caption{Elimination of Bad Trees}
\label{algo:ElimBad}
\begin{algorithmic}[1]
\State Let $F_{{\rm tmp}} = F^*$ 
\While{there is at least one bad tree in 
	$F_{{\rm tmp}}$}
	\State Choose a bad tree $T_b$ in 
		$F_{{\rm tmp}}$ 
	\State Pick an edge $e_{{\rm new}}
		= (v_{{\rm new}}, v_{{\rm new}}')$ 
		from $E^{\dagger}(T_b)$ 
		with the smallest weight, where 
		$E^{\dagger}(T_b) 
		:= \{(v, v') \in E^\dagger
			\ | \ v \in {\cal V}(T_b), 
			v' \notin {\cal V}(T_b) \}$ 
		
	\State Connect $T_b$ to $T_c$ using 
		edge $e_{{\rm new}}$ to create
		a new tree $T_{{\rm new}}$, where $T_c$
		is the tree in $F_{{\rm tmp}}$ which includes
		vertex $v'_{{\rm new}}$ 
	\State Remove $T_b$ and $T_c$ from
		$F_{{\rm tmp}}$, and add the new tree 
		$T_{{\rm new}}$ to $F_{{\rm tmp}}$
\EndWhile
\end{algorithmic}
\end{algorithm}

When we connect a bad tree $T_b$ to $T_c$ 
using edge $e_{{\rm new}}$ in 
Algorithm~\ref{algo:ElimBad} (line 5), the 
possible types of the new tree 
$T_{{\rm new}}$ depend
on the type of tree $T_c$. 

\begin{enumerate}
\item[P1] $T_c$ is a heavy tree -- When we 
connect $T_b$ to a heavy tree, the
total weight of the new tree obviously 
exceeds $\lambda^*$, and the new tree 
$T_{{\rm new}}$ is a heavy tree.

\item[P2] $T_c$ is a light tree -- After
connecting the two trees, $T_{{\rm new}}$ can
be either light or heavy, depending on its
total weight.

\item[P3] $T_c$ is a bad tree -- When two
bad trees are connected, the resulting new
tree $T_{{\rm new}}$ could be any of the 
three types (bad, light or heavy).
\end{enumerate}

Note that a new tree $T_{{\rm new}}$ can be 
a bad tree only if $T_c$ is also a bad tree
(case P3). 

The order in which we choose bad trees in 
Algorithm~\ref{algo:ElimBad} is not important.
In addition, 
Algorithm~\ref{algo:ElimBad} terminates after 
at most $m-1$ rounds; when all $m$ trees
$T_1, \ldots, T_m$ are combined into a single 
spanning tree after at most $m-1$ rounds, 
the resulting tree contains all vertices 
in $V$ and, hence, cannot be a bad tree, 
thereby terminating the algorithm.

Denote by $F_{{\rm opt}}$ the final forest 
$F_{{\rm tmp}}$ produced 
by Algorithm~\ref{algo:ElimBad}.
Let $n_{LT} \in \{0, 1, \ldots, m\}$ be the number 
of light trees in $F_{{\rm opt}}$.

\begin{lemma}
\label{lemma:heacyTreeWeight}
Suppose that $F_h$ is the collection of heavy trees in 
$F_{{\rm opt}}$. Then,
\beqan
w(F_h)
\leq (k - n_{LT})\lambda^*.
\eeqan
\end{lemma}
\prf
A proof of the lemma can be found in 
Appendix~\ref{appen:heacyTreeWeight}.
\QED

{\bf Step 2-ii:} 
Unfortunately, in practice, we do not have 
access to $\lambda^*$ or the optimal cycle 
cover $\mathscr{C}^*$. Hence, 
we cannot determine which trees in the 
forest $F^*$ are bad trees and execute
Algorithm~\ref{algo:ElimBad}. 

For this reason, we consider all possible 
ways in which the $m$ trees in $F^*$ 
can be connected 
to form a new forest, using the edges
in $E^\dagger$. Since $\big| E^{\dagger}
\big| = m-1$, the number of possible
forests we need to consider, including 
the case with a single spanning tree, 
is equal to $2^{m-1}$, and one
of these possible forests coincides 
with $F_{{\rm opt}}$. 
We denote by $\mathscr{F}$
the set of $2^{m-1}$ forests we 
consider, and use $F_{{\rm cand}}$ to
refer to a forest in $\mathscr{F}$.
 
We provide these $2^{m-1}$ forests
in $\mathscr{F}$ as an input to Steps 
3 and 4, one forest at a time.
However, we are primarily interested
in the forest $F_{{\rm opt}}$ 
for finding an 
approximation ratio for the proposed
algorithm.

\subsection{Step 3: Decomposition 
	of Heavy Trees}
\label{stepC}

In addition to the input to the proposed algorithm, 
namely the information regarding the
rooted min-max cycle cover problem, 
Steps 3 and 4 described in this and 
following subsections require another 
parameter $\lambda$. 
The output of these two steps depends 
on whether or not the parameter $\lambda$
is greater than or equal to $\lambda^*$: 
if $\lambda \geq \lambda^*$, 
they return a cycle cover with 
at most $k$ cycles and the maximum cycle
weight less than or equal to $5 \lambda$.

Suppose that the forest under consideration
(out of $2^{m-1}$ possible forests in 
$\mathscr{F}$) is
$F_{{\rm cand}}$, and $\lambda$ is a constant
satisfying $\max_{e \in 
{\cal E}(F_{{\rm cand}})}w(e)
\leq \lambda$. 
We decompose the trees in $F_{{\rm cand}}$ 
with the help of the following lemma.
\\ \vspace{-0.1in}

\begin{lemma}[Lemma 2 of 
\cite{khani2014improved}] \label{lemma:2}
Fix a tree $T$ and $\lambda \geq \max_{e 
\in {\cal E}(T)} w(e)$.
Then, the tree $T$ can be decomposed into 
subtrees $T^s_1, ...,T^s_{m_s}$, such that 
(i) $w(T^s_i) < 2\lambda$ for all $i = 1, 
\ldots, m_s$, 
(ii) ${\cal V}(\cup_{i=1}^{m_s} T^s_i)
= {\cal V}(T)$, and 
(iii) the number of subtrees, $m_s$, 
satisfies
\beqan
m_s \myleq \max \left( \left \lfloor 
	\frac{w(T)}{\lambda}
	\right \rfloor,1 \right).
\eeqan
\end{lemma}

The goal of Step 3 is to construct a new 
forest $\hat{F}_{{\rm cand}}$ that will be 
used in Step 4 to find a rooted cycle cover. 
In particular, when $F_{{\rm cand}} 
= F_{{\rm opt}}$, 
the cycle cover found in Step 4 will possess
a provable upper bound on the maximum
cycle weight. 

To this end, we first
put all trees in $F_{{\rm cand}} $ whose
weight is less than $2 \lambda$ 
in $\hat{F}_{{\rm cand}}$. 
Second, for trees whose weight 
is greater than or equal to $2\lambda$, we 
first decompose them into subtrees as 
described in Lemma~\ref{lemma:2},
using the algorithm proposed in 
\cite{even2004min},
and then put the subtrees in 
$\hat{F}_{{\rm cand}}$. 

The resulting forest $\hat{F}_{{\rm cand}}$ when 
$F_{{\rm cand}} = F_{{\rm opt}}$, which we denote
simply by $\hat{F}_{{\rm opt}}$, has the 
following two properties. Given a 
subgraph $G'$, we define
$w_{\max}(G') := \max_{e \in 
{\cal E}(G')} w(e)$ to be the maximum edge weight in 
the subgraph.
\myskip

\begin{lemma}	\label{lemma:4}
Suppose $\lambda \geq 
w_{\max}(F_{{\rm opt}})$. Then, the maximum weight 
of trees in $\hat{F}_{{\rm opt}}$ is at most $2 
\lambda$.
\end{lemma}
\prf The lemma follows directly from 
Lemma~\ref{lemma:2}
and the construction of $\hat{F}_{{\rm opt}}$.
\QED

\begin{lemma}	\label{lemma:4}
If $\lambda \geq \lambda^*$, the number of 
trees in $\hat{F}_{{\rm opt}}$ is at most $k$.
\end{lemma}
\prf
Please see Appendix~\ref{appen:4} for a proof
of the lemma.
\QED

\subsection{Step 4: Generation of a Rooted 
	Cycle Cover}
\label{stepD}

Recall that, even though every tree in forest
$F_{{\rm cand}}$ contains at least one depot, 
due to the decomposition of trees $T$ with 
$w(T) \geq 2 \lambda$ in Step 3, 
some of the trees in $\hat{F}_{{\rm cand}}$ 
may not include any depot. We partition
the trees in the forest $\hat{F}_{{\rm cand}}$ 
into two forests, $\hat{F}_r$ and 
$\hat{F}_{nr}$. Forest
$\hat{F}_r$ consists of trees that contain 
a depot, and $\hat{F}_{nr}$ comprise trees 
that do not cover any depot. 

We connect each tree $T$ in $\hat{F}_{nr}$ 
to a nearest depot $\tilde{d}(T)$ in 
$\arg\min_{d \in D} \big( \min_{v \in {\cal V}(T)}
w(d, v) \big)$. Note that $\min_{d \in D}
\big( \min_{v \in {\cal V}(T)} w(d, v) \big)
\leq \lambda^*/2$; 
for every vertex in $V^-$, we
have $2 \min_{d \in D} w(v, d) \leq 
\lambda^*$ from the assumption that
the edge weight function $w$ is a metric
and the optimal cycle cover must cover
all vertices in $V^-$.  For each tree $T$ 
in $\hat{F}_{nr}$, we denote the resulting 
tree we obtain after connecting it to a
nearest depot by $T^+$, and let
$\hat{F}^+_{nr} := \{  T^+ | T 
\in \hat{F}_{nr} \}$.  Note that some trees
in $\hat{F}^+_{nr}$ may share a depot with 
other trees in $\hat{F}_r$ or 
$\hat{F}^+_{nr}$. 

\begin{figure}[h]
    \centering
    \begin{subfigure}[t]{0.23\textwidth}
        \centering
        \usetikzlibrary{arrows}
        \begin{tikzpicture}[scale=0.65, every node/.style={transform shape}]
            \fill[color=black] (-0.5,1) node (v1) {} circle (.2);
            \fill[color=black] (-3.5,-1) node (v3) {} circle (.2);
            \fill[color=black] (-1,-1) node (v2) {} circle (.2);
            \fill[color=red] (0.5,-2.5) node (v4) {} circle (.2);
            \draw (v1) -- (v2);
            \draw (v3) -- (v2);
            \draw (v2) -- (v4);
            \draw [-latex](0,-2.5) .. controls (0.5,-3.5) and (1.5,-2.5) .. (0.5,-2);
            \draw [-latex](0,-1.5) .. controls (-0.5,-1) and (-0.5,-1) .. (-0.5,0);
            \draw [-latex](-0.5,0.5) .. controls (0.5,1.5) and (-1,2) .. (-1,0.5);
            \draw [-latex](-1,0) .. controls (-1,-0.5) and (-2,-0.5) .. (-2.5,-0.5);
            \draw [-latex](-3,-0.5) .. controls (-4.5,-0.5) and (-4.5,-1.5) .. (-3,-1.5);
            \draw [-latex](-2.5,-1.5) .. controls (-1.5,-1.5) and (-1,-1.5) .. (-0.5,-2);
        \end{tikzpicture}
        \caption{an Eulerian cycle of a tree}
    \end{subfigure}
    \begin{subfigure}[t]{0.23\textwidth}
        \centering
        \usetikzlibrary{arrows}
        \begin{tikzpicture}[scale=0.65, every node/.style={transform shape}]
            \fill[color=black] (-0.5,1) node (v1) {} circle (.2);
            \fill[color=black] (-3.5,-1) node (v3) {} circle (.2);
            \fill[color=black] (-1,-1) node (v2) {} circle (.2);
            \fill[color=red] (0.5,-2.5) node (v4) {} circle (.2);
            \draw [-latex](0.5,-2.2) -- (-0.7,-1);
            \draw [-latex](-0.7,-1) -- (-0.2,1);
            \draw [-latex](-0.8,1) -- (-3.5,-0.7);
            \draw [-latex](-3.5,-1.3) -- (0.2,-2.5);
        \end{tikzpicture}
        \caption{a cycle after short cutting}
    \end{subfigure}
    \caption{An example of generating a cycle for a tree}
    \label{fig:eulerianTour}
\end{figure}
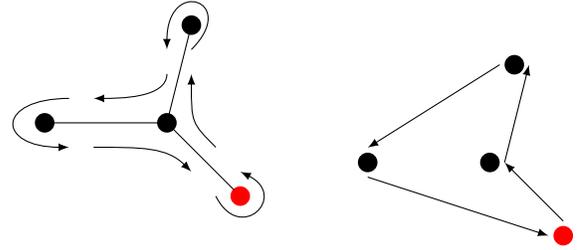

Let us consider the case when 
$\hat{F}_{{\rm cand}} = \hat{F}_{{\rm opt}}$. 
Recall from Lemma~\ref{lemma:4} that 
the number of trees in 
$\hat{F}^+ := \hat{F}_r \cup 
\hat{F}^+_{nr}$ does not exceed $k$ if 
$\lambda \geq \lambda^*$. 
Each tree in $\hat{F}^+$ contains one 
depot, and its weight 
is at most $2\lambda + \lambda^*/2$ from
its construction.
By finding the Eulerian cycle and performing
short cutting (as shown in Fig.
\ref{fig:eulerianTour})
for each tree in $\hat{F}^+$, 
we can find a cycle cover
$\mathscr{C}_{{\rm cand}}$. When 
$\hat{F}_{{\rm cand}} = \hat{F}_{{\rm opt}}$, 
the cycle cover
$\mathscr{C}_{{\rm cand}}$ has the maximum 
cycle weight of at most $4\lambda + \lambda^*$.
\myskip

\begin{lemma} \label{lemma:5}
When $F_{{\rm cand}} = F_{{\rm opt}}$, 
for $\lambda \geq w_{\max}(F_{{\rm opt}})$, 
Steps 3 and 4 together generate a cycle 
cover with the maximum cycle weight of 
at most $4 \lambda + \lambda^*$. 
In addition, if $\lambda \geq \lambda^*$, 
there are at most $k$ cycles in the cycle 
cover.
\end{lemma}

\section{The Proposed Algorithm}
	\label{sec:Complete}
	
We are now ready to present our proposed
algorithm for \eqref{eq:objective}. The 
pseudocode of the proposed algorithm is 
shown in Algorithm~\ref{alg:complete}. 
We denote the maximum cycle weight 
of a cycle cover produced in Step 4
(line 12), 
namely $\mathscr{C}_{{\rm cand}}$, 
by $O_{{\rm cand}}$. It is clear from 
the pseudocode, if the number of trees
in $\hat{F}_{{\rm cand}}$ is larger than 
$k$, we do not need to execute Step 4 as
it will not yield a feasible solution. 

\begin{algorithm}[h]
  \caption{ Rooted Min-Max Cycle Cover Algorithm}
  \label{alg:complete}
\begin{algorithmic}[1]
    \Require
 	{(i) a complete undirected graph $G=(V,E)$;
     (ii) a weight function $w:E \rightarrow \R_+$; 
     (iii) depot set $D \subseteq V$; 
     (iv) maximum number of cycles $k$; 
     (v) constant $\epsilon \in (0, 1)$}
    \Ensure
      {A cycle cover with the maximum weight 
      of cycles less than 
	  $(5+\epsilon)\lambda^*$} 
	  
    \State {Construct a rooted spanning forest $F^*$
    	\Comment{(Step 1)}} 
    \State {Find the edge set $E^\dagger$ 
    	\Comment{(Step 2-i)}} 
    
    \State {Generate $\mathscr{F}$ with $2^{m-1}$
    	possible forests}
    
    \State {Set $\mathscr{C}_{{\rm tmp}} = 
    \emptyset$ and $O_{{\rm tmp}} = \infty$}
    
    \For{each forest $F_{{\rm cand}}$ in 
    	$\mathscr{F}$} 
    	\Comment{(Step 2-ii)}
    		\State $\ell \gets 1$
    		\State {$a_\ell \gets 
    			w_{\max}(F_{{\rm cand}})$ 
    			and $b_\ell \gets (n+k) 
    				w_{\max}(G)$}
    		\While{true}
    			\State {$\lambda_\ell \gets 0.5 
    				(a_\ell + b_\ell)$}
    			
    			\State Construct forest 
        			$\hat{F}_{{\rm cand}}$
        			with $\lambda = \lambda_\ell$
        			\Comment{(Step 3)}
        			
        		\If{number of trees in 
        			$\hat{F}_{{\rm cand}}
        			\leq k$}
        			
        			\State Find a cycle cover 
        				$\mathscr{C}_{{\rm cand}}$
        				\Comment{(Step 4)}        		
    				\If{$O_{{\rm cand}}
    					< O_{{\rm tmp}}$}
        				\State $\mathscr{C}_{{\rm tmp}}
        				\gets \mathscr{C}_{{\rm cand}}$
        			
        				\State $O_{{\rm tmp}} \gets
        					O_{{\rm cand}}$
        			\EndIf
        			\State $a_{\ell+1} 
        				\gets a_\ell$ and 
        				$b_{\ell+1} \gets
        				\lambda_\ell$ 
        		\Else 
        			\State $a_{\ell + 1} 
    					\gets \lambda_\ell$
    					and $b_{\ell+1}
    					\gets b_\ell$
        		\EndIf
    			\If{$b_\ell - a_\ell < 0.5 
    				\epsilon \ a_\ell$}
    				\State Break \Comment{break out
    					of \texttt{while} loop}
    			\Else
    				\State $\ell \gets \ell + 1$
    			\EndIf
        	\EndWhile
    \EndFor
\end{algorithmic}
\end{algorithm}

Given a forest $F_{{\rm cand}}$
in $\mathscr{F}$, a binary search for a 
suitable value of $\lambda$ (lines 8 - 26) 
is performed 
over the interval $[w_{\max}(F_{{\rm cand}}), 
(n+k) w_{\max}(G)]$.
The goal of the binary search
is not to find the smallest value of
$\lambda$ greater than or equal to 
$\lambda^*$. Instead, it is to find a value
of $\lambda$ for which Step 3 produces a 
feasible forest $\hat{F}_{{\rm cand}}$ 
with at most 
$k$ trees (line 11) and the termination 
condition in line 21 is satisfied. 
 
Lemma~\ref{lemma:4} guarantees that our 
algorithm will produce at least one 
feasible solution with at most $k$ cycles 
as long as $\lambda^*$ lies in $\big[
w_{\max}(F_{{\rm opt}}), 
(n+k) w_{\max}(G) \big)$. Moreover, 
Lemma~\ref{lemma:5} ensures that 
the maximum cycle weight of the cycle 
cover found by the algorithm is at most 
$4\lambda + \lambda^*$. Therefore, by 
selecting a suitable value of $\lambda$, 
we can prove the approximation ratio of
$5 + \epsilon$ for the algorithm, 
where $\epsilon$ is a constant in the interval
(0, 1), which is an input to the algorithm we
select. 
\myskip

\begin{theorem}	\label{thm:1}
For fixed $\epsilon \in (0, 1)$, Algorithm 
\ref{alg:complete} returns a rooted cycle 
cover with at most $k$ cycles and 
the maximum cycle weight less than or equal 
to $(5 + \epsilon) \lambda^*$. 
\end{theorem}
\prf
A proof of the theorem can be found in Appendix~\ref{appen:thm1}. 
\QED



\section{Complexity of the Proposed Algorithm}
	\label{sec:Complexity}
	
In this section, we study the 
complexity of the proposed algorithm 
(Algorithm~\ref{alg:complete}). For a fixed
number of depots, $m$, we will show that 
the overall complexity is $O(n^2)$, where
$n$ is the number of vertices in $V$, by 
examining the computational requirements of
each step. 
\myskip

\begin{itemize}
\item{\bf Step 1:} The Prim's 
algorithm for finding a minimum spanning
tree $T'$ (Step 1-i) has 
complexity of $O(n^2)$. 
Step 1-ii simply requires finding the
closest depot for each vertex in $V^-$
and has complexity of $O(n)$. 

\item{\bf Step 2:} A naive way to find
the edge set $E^\dagger$ is to construct
a minimum spanning tree for $G$
(using, for instance, Prim's algorithm
\cite{Algorithm}) and remove the edges
in $F^*$. This has complexity $O(n^2)$. 

\item{\bf Step 3:} For every 
forest $F_{{\rm cand}}$ in $\mathscr{F}$,
$w_{\max}(F_{{\rm cand}}) \geq 
w_{\max}(F^*)$. Hence, the \texttt{while} loop
in lines 8 through 26 terminates after
at most 
\beqan
&& \myhb \log_2\left( \frac{(n+k) w_{\max}(G)}
{0.5 \epsilon w_{\max}(F^*)} \right) \lb
\myeq O\big( \log_2(n+k) \big)
+ O\big( \log_2(1/\epsilon) \big)  \lb
&& + O\big(\log_2( w_{\max}(G) / w_{\max}(F^*)) \big)
\eeqan
iterations. 
In addition, decomposing trees whose
weight is greater than or equal to
$2 \lambda$ using `splitting' 
described in \cite{even2004min, 
khani2014improved} has complexity 
of $O(n)$. 

\item{\bf Step 4:} 
Both connecting trees in $\hat{F}_{nr}$
to the closest depots and finding 
the Eulerian cycles and associated
rooted cycle cover have
complexity of $O(n)$. 

\end{itemize}

From the above discussion, the overall 
complexity of the proposed algorithm is 
$O\big( n^2 \big) 
+ O\big( 2^{m-1} n \log(n+k) \big)
+ O\big(n  \log(1/\epsilon) \big)$, 
provided that  $w_{\max}(G) / w_{\max}(F^*)
= O\big( \exp(n) \big)$.
In practice, we have $k = O(n)$ and, 
as a result, the complexity
is in fact $O(n^2) + O\big( 2^{m-1} n \log(n)
\big) + O\big(n  \log(1/\epsilon) \big)$. 
Therefore, when both the number of depots 
($m$) and $\epsilon$ are fixed, the complexity 
is $O(n^2)$. Furthermore, since the $2^{m-1}$
possible forests in $\mathscr{F}$ 
can be considered independently, 
the execution of the \texttt{for} loop in 
the algorithm (lines 5 - 27) can be 
parallelized.  

Finally, recall that $F^*$ is a spanning forest
consisting of $m$ trees, which are constructed 
from the minimum spanning tree $T'$ of $G'$ in 
Step 1-i and are rooted at the $m$ depots. 
Hence, unless \underline{all} 
remaining $n-m$ vertices 
get closer to the $m$ depots with increasing $n$, 
the assumption $w_{\max}(G) / w_{\max}(F^*)
= O\big( \exp(n) \big)$ will hold. For instance, 
it is well known that the longest edge of 
the minimum spanning tree covering $n$
independent and identically distributed 
points in a unit ball in $\R^d$, 
$d \geq 2$, converges
almost surely to $c \big( \log(n) / n \big)^{1/d}$
as $n$ goes to $\infty$, where $c$ is
some constant
\cite{Penrose99_SL_LongestEdge}. Hence, 
for a random geometric graph $G$ with
$d = 3$, we have $w_{\max}(G) / w_{\max}(F^*)
= O\big( (n / \log(n) )^{1/3} \big)$.

\section{Conclusion}

We studied the problem of robot deployment planning
for inspection of civil structures, with an emphasis on 
bridges. Each mobile robot is stationed at a depot where it
recharges the battery and is tasked with the inspection 
of a subset of points or segments of the structure. 
The problem is formulated as a (variant of the) 
uncapacitated rooted min-max cycle cover problem.
We proposed a new algorithm with approximation ratio 
of $5 + \epsilon$.

Our formulation as an uncapacitated rooted min-max 
cycle cover problem assumes that depots are large 
enough to house as many robots as needed. We suspect
that any good solution will distribute the robots
evenly across a bridge to minimize the max cycle 
weight. In some cases, however, the depots may have
limited slots for the robots for housing and 
recharging. This may require reformulating the 
problem as a capacitated rooted min-max cycle
cover, similar to that studied by the authors
of \cite{xu2015approximation}. We are currently
working on this problem. 

\section*{Acknowledgment}
This work was supported in part by an NSF grant 
ECCS 1446785 and a NIST grant 70NANB16H024.

\begin{appendices}

\section{Proof of Lemma~\ref{lemma:heacyTreeWeight}}
	\label{appen:heacyTreeWeight}

Let $V_h := {\cal V}(F_h)$ be the set of vertices in $F_h$
and $F^*_h$ the smallest subset of trees in $F^*$, which 
covers $V_h$, i.e., $V_h = {\cal V}(F^*_h)$. 
Similarly, we define $\mathscr{C}^*_h$ to be the
smallest subset of cycles in $\mathscr{C}^*$ such that
the union of vertices in the cycles contain $V_h$, i.e., 
$V_h \subseteq \cup_{C \in \mathscr{C}_h^*}
{\cal V}(C) =: V_{\mathscr{C}^*_h}$.  

First, note that, by deleting an edge with the largest
weight in each cycle $C \in \mathscr{C}_h^*$, we can obtain a 
rooted forest $F_{\mathscr{C}_h^*}$ that covers 
$V_{\mathscr{C}_h^*}$ and each tree in the forest 
contains a 
depot in $D$ because each cycle $C \in \mathscr{C}_h^*$
must be rooted at a depot. Moreover, from the construction 
of $F^*$ in Step 1 (subsection~\ref{stepA}), 
$F^*_h$ is a minimum weight rooted spanning 
forest of $V_h$. 
Therefore, 
because $V_h \subseteq V_{\mathscr{C}_h^*}$, we have
\beqan
w(F^*_h) \myleq w(F_{\mathscr{C}_h^*})\\
\myeq w(\mathscr{C}_h^*)-\sum_{C \in \mathscr{C}_h^*} 
	w_{\max}(C) \\
\myleq (k - n_{LT})\lambda^*-\sum_{C \in \mathscr{C}_h^*} 
	w_{\max}(C),
\eeqan
where $w_{\max}(G') := \max_{e \in {\cal E}(G')} w(e)$
for any subgraph $G'$, 
and the second inequality follows from the fact that 
(i) there are at most $k - n_{LT}$ cycles in $\mathscr{C}_h^*$ 
because at least $n_{LT}$ cycles are covered by 
$n_{LT}$ light 
trees and (ii) for every cycle $C$ in $\mathscr{C}^*_h$, 
$w(C) \leq \lambda^*$. 

In order to bound the total weight of the trees
in $F_h$, in addition to the bound for $w(F_h^*)$, 
we need to account for the weights of 
the edges that were introduced in the the trees
of $F_h$ through the merging process of bad
trees in $F_h^*$. For this, we consider two cases 
based on how the heavy tree $T_h$ in $F_h$, which
does not belong to $F^*_h$, was created. 
Let $T_b$ be the bad tree that was connected
to another tree to form $T_h$. 

\benum
\item[C1.] The bad tree $T_b$ is a bad tree 
in $F^*_h$. 

\item[C2.] The bad tree $T_b$ is not a bad tree
in $F^*_h$. In this case, $T_b$ must have
been created as a result of merging at 
least two bad trees in previous round(s)
of connecting bad trees. 

\eenum

{\bf Case C1:} Let $d$ be the depot covered 
by $T_b$, and choose a cycle $C^*_d$ in 
$\mathscr{C}_h^*$,
which includes $d$. The existence of such a
cycle is guaranteed because the cycles in 
$\mathscr{C}_h^*$ must cover all vertices
in $V_h$, including depots in $V_h$.
 
Since $T_b$ is a bad tree, there exists 
at least one vertex $v^*$ in $C^*_d$ which 
is not covered by $T_b$ and instead belongs
to a different tree; otherwise, $T_b$
would be a light tree. As a result, 
there is at least one edge $e'$ that connects 
a vertex in $T_b$ to another tree 
and $e'$ belongs to $C^*_d$. Since Algorithm
\ref{algo:ST} uses an edge with the smallest 
weight to connect $T_b$ to the rest of
spanning tree in the process of finding 
the edge set $E^\dagger$, the edge 
$e^+ \in E^\dagger$ used to connect $T_b$ 
to another tree (hence, is a new edge 
introduced in $T_h$) satisfies
\beqa
w(e^+) 
\leq w(e') 
\leq w_{\max}(C^*_d).
	\label{eq:c1-1}
\eeqa

{\bf Case C2:} We first consider the 
construction of the bad tree $T_b$. Since $T_b$
is a bad tree, it must have been produced
as a result of connecting two or more bad 
trees in $F^*_h$. 
Assume that $T_b$ is constructed as a result
of $L$ rounds of merging bad trees, and 
for each $\ell = 1, \ldots, L$, let $T_{b_\ell}$
be the tree chosen to be connected to
another tree $T'_{b_\ell}$ for the $\ell$-th
round of merging. Furthermore, for
each $\ell = 1, \ldots, L$, denote
the edge chosen to connect the two bad trees
$T_{b_\ell}$ and $T'_{b_\ell}$ 
by $e^+_\ell \in E^\dagger$. 

For each $\ell = 1, \ldots, L$, suppose that 
$d_\ell$ is a depot in $T_{b_\ell}$ which has 
not been considered in the previous
rounds, i.e., $d_\ell \notin
\{d_1, \ldots, d_{\ell-1}\}$, and a cycle $C^*_\ell$
in $\mathscr{C}^*_h$ covers $d_\ell$. 
We can always find a depot $d_\ell$ that meets
the above condition because $T_{b_\ell}$
contains at least $r+1$ depots, where 
$r$ is the number of rounds of merging
process needed to form
$T_{b_\ell}$ before (hence, $T_{b_\ell}$ 
includes at least $r+1$ bad trees in 
$F^*_h$) and 
exactly $r$ of these depots were 
considered in the $r$ rounds. 
This also implies that 
$C^*_\ell \neq C^*_l$ for all $\l = 1, 
\ldots, \ell-1$, because each cycle
in $\mathscr{C}^*_h$ contains exactly one 
depot. 

Following the same argument used in the 
previous case (C1), for each $\ell = 1, 
\ldots, L$, we can find an edge $e'_\ell$
in the cycle $C^*_\ell$, which  
connects a vertex in 
$T_{b_\ell}$ to another tree. Consequently, 
\beqa
w(e^+_\ell) 
\leq w(e'_\ell)
\leq w_{\max}(C^*_\ell).
	\label{eq:c2-1}
\eeqa

Finally, the bad tree $T_b$ is 
connected to another tree to form 
the heavy tree $T_h$. By the same
argument used in case C1, we can find
a depot $d$ in $T_b$ such that 
$d \notin \{d_1, \ldots, d_L\}$ 
and a cycle
$C^*_d$ containing $d$ so that 
the weight of the edge $e^+$ that
connects $T_b$ to another tree
to form $T_h$ satisfies
\begin{equation}
w(e^+) 
\leq w(e') 
\leq w_{\max}(C^*_d), 
	\tag{\ref{eq:c1-1}}
\end{equation}
where $e'$ is an edge in $C^*_d$ which
connects $T_b$ to another tree
as described in case C1. 

Note from the above discussion that, for each added 
edge $e^\star$ during the construction of a heavy 
tree 
$T_h$ that is not in $F^*_h$, we can find a distinct 
cycle in $\mathscr{C}^*_h$ such that the weight of 
the new added edge is upper bounded by the 
maximum edge weight of the cycle.  
Let $E^+ (\subseteq E^\dagger)$ denote the 
set of edges that were added
to connect trees in $F^*_h$ to create $F_h$.
Then, we have $\sum_{e \in E^+} w(e) 
\leq \sum_{C \in \mathscr{C}^*_h} w_{\max}(C)$ 
because no cycle $C \in \mathscr{C}^*_h$ is
considered more than once during the process. 
Therefore, we have 
\beqa
&& \myhb w(F_h)
= w(F_h^*)+\sum_{e \in E^+} 
	w(e) \lb
\myleq (k - n_{LT})\lambda^*
	- \sum_{C \in \mathscr{C}_h^*} w_{\max}(C)
	+ \sum_{e \in E^+}  w(e) \lb
\myleq (k - n_{LT})\lambda^*
	- \sum_{C \in \mathscr{C}_h^*} w_{\max}(C) 
	+ \sum_{C \in \mathscr{C}_h^*} 
		w_{\max}(C) \lb
\myleq (k - n_{LT})\lambda^*.
	\nonumber
\eeqa

\section{Proof of Lemma~\ref{lemma:4}}
	\label{appen:4}

First, note that all $n_{LT}$ light trees 
in $F_{{\rm opt}}$ belong to 
$\hat{F}_{{\rm opt}}$ 
since the weight of light trees is less 
than $\lambda^*$, which is less than 
or equal to $\lambda$ by assumption.

Order the heavy trees in $F_h$ by 
decreasing weight: $T_1, 
\ldots, T_{n_1}, T_{n_1 + 1}, 
\ldots, T_{n_1 + n_2}$, where
$n_1$ is the number of heavy trees
whose weight is greater than or
equal to $2 \lambda$, i.e., 
$w(T_i) \geq 2\lambda$ for $i = 1, 
\ldots, n_1$, and $\lambda^* \leq
w(T_i) < 2 \lambda$ for $i = n_1+1,
\ldots, n_1 + n_2$. 
Let $\hat{F}_h$ be the set of trees
in $\hat{F}_{{\rm opt}}$ which come
from the heavy trees in $F_h$. 

Using Lemma~\ref{lemma:2}, we can 
upper bound $|\hat{F}_h|$ as
follows.
\beqan
|\hat{F_h}| 
\myleq \sum_{i=1}^{n_1} \max \left( 
	\left\lfloor \frac{w(T_i)}{\lambda} 
	\right\rfloor, 1 \right)
	+ n_2 \\
\myleq \sum_{i=1}^{n_1}
	\frac{w(T_i)}{\lambda} + n_2 \\
\myeq \frac{w(F_h) 
	- \sum_{i = n_1 + 1}^{n_1 + n_2}w(T_i)}
	{\lambda} + n_2 \\
\myleq \frac{w(F_h)- n_2 \lambda^*}{\lambda}
	+ n_2,
\eeqan
where the last inequality follows from 
$\lambda^* \leq w(T_i)$ for all $i = n_1 + 1, 
\ldots, n_1 + n_2$. 
Using the bound for $w(F_h)$ in Lemma
\ref{lemma:heacyTreeWeight}, we obtain
\beqan
|\hat{F_h}| 
\myleq \frac{(k - n_{LT})\lambda^* 
	- n_2 \lambda^*}{\lambda} + n_2 \\
\myleq (k - n_{LT})- n_2 + n_2 \\
\myeq k - n_{LT}, 
\eeqan
where $n_{LT}$ is the number of light trees 
in $F_{{\rm opt}}$, and the second inequality is
a consequence of the assumption $\lambda^*
\leq \lambda$.
Thus, the number of trees in 
$\hat{F}_{{\rm opt}}$ can 
be upper bounded by $k$ because 
$|\hat{F}_{{\rm opt}}|
= |\hat{F_h}| + n_{LT} \leq k$.

\section{Proof of Theorem~\ref{thm:1}}
	\label{appen:thm1}

In order to prove the theorem, we only need
to show that at least one of the $2^{m-1}$ 
forests in $\mathscr{F}$ leads to a solution 
that satisfies
the approximation ratio in the theorem. 
To this end, we consider the forest 
$F_{{\rm opt}}$ generated from 
the bad tree elimination process 
(Algorithm~\ref{algo:ElimBad}) with an
optimal cycle cover $\mathscr{C}^*$. Recall
that $F_{{\rm opt}}$ always belongs to 
$\mathscr{F}$. 

For forest $F_{{\rm opt}}$, the binary search 
in Algorithm~\ref{alg:complete} is performed 
over the interval $[w_{\max}(F_{{\rm opt}}), 
(n+k) w_{\max}(G)]$.
It is clear that $(n + k) w_{\max}(G) 
> \lambda^*$ since the optimal cycle 
cover cannot have more than $(n+k-m)$ edges. 
The following lemma demonstrates that 
$w_{\max}(F_{{\rm opt}})$ is a lower
bound on the optimal value $\lambda^*$.

\begin{lemma}	\label{lemma:6}
The optimal value $\lambda^*$ is greater
than or equal to $w_{\max}(F_{{\rm opt}})$.
\end{lemma}
\begin{proof}
Let $e^* := (u,v) \in \arg\max_{e \in 
{\cal E}(F_{{\rm opt}})} w(e)$.

We shall consider following two cases.
\begin{enumerate}
\item[C-i] $e^* \in {\cal E}(F^*)$ -- 
Suppose that $e^*$ belongs to a tree 
$T^*$ in $F^*$. Without loss of 
generality, we assume that, when we 
remove $e^*$ from $T^*$ and divide 
it into two 
subtrees, vertex $v$ belongs to the 
subtree with the depot in $T^*$.

Suppose that the claim is false and 
$w(e^*) > \lambda^*$. Let $C_u$ be a  
cycle in the optimal cycle cover
$\mathscr{C}^*$ which covers vertex 
$u$ and $d_u$ be the depot in the cycle
$C_u$. Denote by $e_u$ the edge $(u, d_u)$.
Note that $w(e_u) \leq 
\lambda^*/2$ 
because $2 w(e_u) \leq w(C_u) 
\leq \lambda^*$ because $w$ 
satisfies the triangle inequality. 

Consider a new spanning forest $F'$ we 
can generate by replacing the edge $e^*$
in $F^*$ with the edge $e_u$. It 
is clear that $F'$ is also a rooted 
spanning forest of $G$ since every  
tree in $F'$ includes exactly one depot 
and $F'$ covers all vertices in $V$. 
However, 
\beqan
w(F') 
\myeq w(F^*) - w(e^*) + w(e_u)
< w(F^*),
\eeqan
because $w(e_u) < \lambda^* < w(e^*)$. 
This contradicts the earlier assumption 
that $F^*$ is a minimum rooted spanning 
forest.

\item[C-ii] $e^* \notin {\cal E}(F^*)$ -- 
This means that $e^*$ is an edge added 
by Algorithm~\ref{algo:ElimBad} during 
the bad tree elimination process. From 
\eqref{eq:c1-1} and \eqref{eq:c2-1}, 
\beqan
w(e^*) 
\leq \max_{\ell = 1, \ldots, q^*}
	w_{\max}(C^*_\ell)
\leq \max_{\ell = 1, \ldots, q^*}
	w(C^*_\ell) 
\leq \lambda^*. 
\eeqan

\end{enumerate}
Since the claim of the lemma holds
in both cases, this completes the proof
of the lemma. 
\end{proof}



An important observation is that, during
the binary search in Algorithm
\ref{alg:complete}, the
lower bound of the interval, $a_\ell$, 
increases for the next round only if 
the returned forest $F_{{\rm cand}}$ 
in the $\ell$-th round
has more than $k$ trees (line 19). 
However, since Steps 3 and 4 guarantee
a rooted cycle cover with at most $k$ 
cycles when $\lambda \geq \lambda^*$ 
(Lemma~\ref{lemma:5}) 
and $a_1 \leq \lambda^*$ by 
Lemma~\ref{lemma:6}, we have $a_\ell
\leq \lambda^*$ for all $\ell = 1, 2, 
\ldots$.

Suppose that the process terminates after 
$N$ rounds, i.e., $b_N - a_N < \epsilon \ 
a_{N} / 2$. We denote the returned rooted 
cycle cover after the $N$-th
round by $\mathscr{C}^\star$. Using 
Lemma~\ref{lemma:5}, we can show
that the weight of the cycle cover 
$\mathscr{C}^\star$ is 
upper bounded by $(5+\epsilon) \lambda^*$ 
as follows.
\beqan
w(\mathscr{C}^\star) 
\myleq 4 \lambda_N + \lambda^* 
= 4 (a_N + b_N)/2 + \lambda^* \\
\myl 2(a_N + a_N + \epsilon \ a_N/2) 
	+ \lambda^* \\
\myleq (5 + \epsilon) \lambda^*
\eeqan
The last inequality is a consequence
of the earlier observation that $a_\ell
\leq \lambda^*$ for all $\ell = 1, 2, \ldots, 
N$. Obviously, the maximum cycle weight of 
$\mathscr{C}^\star$ is 
upper bounded by $w(\mathscr{C}^\star)$. 
Therefore, this proves the theorem. 


\end{appendices}



%
\bibliography{root}{}

\begin{thebibliography}{10}

\bibitem{FHWA}
\url{https://www.fhwa.dot.gov/bridge/nbi/no10/condition18.cfm}.

\bibitem{even2004min}
G.~Even, N.~Garg, J.~K{\"o}nemann, R.~Ravi, and A.~Sinha, ``Min--max tree
  covers of graphs,'' {\em Operations Research Letters}, vol.~32, no.~4,
  pp.~309--315, 2004.

\bibitem{arkin2006approximations}
E.~M. Arkin, R.~Hassin, and A.~Levin, ``Approximations for minimum and min-max
  vehicle routing problems,'' {\em Journal of Algorithms}, vol.~59, no.~1,
  pp.~1--18, 2006.

\bibitem{khani2014improved}
M.~R. Khani and M.~R. Salavatipour, ``Improved approximation algorithms for the
  min-max tree cover and bounded tree cover problems,'' {\em Algorithmica},
  vol.~69, no.~2, pp.~443--460, 2014.

\bibitem{jorati2013approximation}
A.~Jorati, {\em Approximation algorithms for some min-max vehicle routing
  problems}.
\newblock University of Alberta (Canada), 2013.

\bibitem{xu2015approximation}
W.~Xu, W.~Liang, and X.~Lin, ``Approximation algorithms for min-max cycle cover
  problems,'' {\em IEEE Transactions on Computers}, vol.~64, no.~3,
  pp.~600--613, 2015.

\bibitem{yu2016improved}
W.~Yu and Z.~Liu, ``Improved approximation algorithms for some min-max and
  minimum cycle cover problems,'' {\em Theoretical Computer Science}, vol.~654,
  pp.~45--58, 2016.

\bibitem{christofides1976worst}
N.~Christofides, ``Worst-case analysis of a new heuristic for the travelling
  salesman problem,'' tech. rep., Carnegie-Mellon Univ Pittsburgh Pa Management
  Sciences Research Group, 1976.

\bibitem{frederickson1976approximation}
G.~N. Frederickson, M.~S. Hecht, and C.~E. Kim, ``Approximation algorithms for
  some routing problems,'' in {\em Foundations of Computer Science, 1976., 17th
  Annual Symposium on}, pp.~216--227, IEEE, 1976.

\bibitem{farbstein2015min}
B.~Farbstein and A.~Levin, ``Min--max cover of a graph with a small number of
  parts,'' {\em Discrete Optimization}, vol.~16, pp.~51--61, 2015.

\bibitem{nagamochi2005approximating}
H.~Nagamochi, ``Approximating the minmax rooted-subtree cover problem,'' {\em
  IEICE Transactions on Fundamentals of Electronics, Communications and
  Computer Sciences}, vol.~88, no.~5, pp.~1335--1338, 2005.

\bibitem{nagamochi2007approximating}
H.~Nagamochi and K.~Okada, ``Approximating the minmax rooted-tree cover in a
  tree,'' {\em Information Processing Letters}, vol.~104, no.~5, pp.~173--178,
  2007.

\bibitem{karakawa2009minmax}
S.~Karakawa, E.~Morsy, and H.~Nagamochi, ``Minmax tree cover in the euclidean
  space,'' in {\em International Workshop on Algorithms and Computation},
  pp.~202--213, Springer, 2009.

\bibitem{xu2010approximation}
Z.~Xu and Q.~Wen, ``Approximation hardness of min--max tree covers,'' {\em
  Operations Research Letters}, vol.~38, no.~3, pp.~169--173, 2010.

\bibitem{xu2012approximation}
Z.~Xu, D.~Xu, and W.~Zhu, ``Approximation results for a min--max
  location-routing problem,'' {\em Discrete Applied Mathematics}, vol.~160,
  no.~3, pp.~306--320, 2012.

\bibitem{Algorithm}
T.~H. Cormen, C.~E. Leiserson, R.~L. Rivest, and C.~Stein, {\em Introduction to
  Algorithms}.
\newblock McGraw-Hill, second~ed., 2001.

\bibitem{Penrose99_SL_LongestEdge}
M.~D. Penrose, ``A strong law for the longest edge of the minimal spanning
  tree,'' {\em The Annals of Probability}, vol.~27, no.~1, pp.~246--260, 1999.

\end{thebibliography}
\bibliographystyle{ieeetr}

\end{document}